\documentclass[journal]{IEEEtran}
\IEEEoverridecommandlockouts
\usepackage{cite}
\usepackage{amsmath,amssymb,amsfonts,amsthm}
\usepackage{algorithmic}
\usepackage{graphicx}
\usepackage{textcomp}
\usepackage{xcolor}
\usepackage{booktabs}
\usepackage{glossaries}

\newtheorem{theorem}{Theorem}
\newtheorem{lemma}{Lemma}

\newtheorem{definition}{Definition}

\newtheorem{proposition}{Proposition}

\def\BibTeX{{\rm B\kern-.05em{\sc i\kern-.025em b}\kern-.08em
    T\kern-.1667em\lower.7ex\hbox{E}\kern-.125emX}}

\newacronym{csi}{CSI}{channel state information}
\newacronym{cqi}{CQI}{channel quality indicator}
\newacronym{ack}{ACK}{acknowledgement}
\newacronym{arq}{ARQ}{automatic repeat request}
\newacronym{awgn}{AWGN}{additive white Gaussian noise}
\newacronym{cc}{CC}{chase combining}
\newacronym{mmw}{mmWave}{millimeter-wave}
\newacronym{dp}{DP}{dynamic programming}
\newacronym{fec}{FEC}{forward error correction}
\newacronym{harq}{HARQ}{hybrid automatic repeat request}
\newacronym{hspa}{HSPA}{high speed packet access}
\newacronym{per}{PER}{packet error rate}
\newacronym{cper}{CPER}{consecutive packet error rate}
\newacronym{iid}{i.i.d.}{independent and identically distributed}
\newacronym{ir}{IR}{incremental redundancy}
\newacronym{lte}{LTE}{long term evolution}
\newacronym{mdp}{MDP}{markov decision process}
\newacronym{mrc}{MRC}{maximal-ratio combining}
\newacronym{nack}{NAK}{negative acknowledgement}
\newacronym{pdf}{pdf}{probability density function}
\newacronym{wimax}{WiMax}{worldwide interoperability for microwave access}
\newacronym{3gpp}{3GPP}{3rd generation partnership project}
\newacronym{ofdm}{OFDM}{orthogonal frequency-division multiplexing}
\newacronym{ofdma}{OFDMA}{orthogonal frequency-division multiple access}
\newacronym{wlan}{WLAN}{wireless local area network}
\newacronym{mse}{MSE}{mean-squared error}
\newacronym{mmse}{MMSE}{minimum mean-squared error}
\newacronym{gsm}{GSM}{global system for mobile communications}
\newacronym{edge}{EDGE}{enhanced data \gls{gsm} environment}
\newacronym{stbc}{STBC}{space-time block code}
\newacronym{amc}{AMC}{adaptive modulation and coding}
\newacronym{snr}{SNR}{signal to noise ratio}
\newacronym{sinr}{SINR}{signal to interference and noise ratio}
\newacronym{mi}{MI}{mutual information}
\newacronym{acmi}{ACMI}{accumulated mutual information}
\newacronym{nacmi}{NACMI}{normalized ACMI}
\newacronym{cdi}{CDI}{channel distribution information}
\newacronym{latr}{LATR}{long-term average transmission rate}
\newacronym{rtr}{RTR}{round transmission rate}
\newacronym{pomdp}{POMDP}{Partially Observable Markov Decision Process}
\newacronym{fd}{FD}{full-duplex}
\newacronym{hd}{HD}{half-duplex}
\newacronym{td}{TD}{time division}
\newacronym{la}{LA}{link adaptation}
\newacronym{tdma}{TDMA}{time-division multiple access}
\newacronym{mac}{MAC}{Media Access Control}
\newacronym{uwb}{UWB}{Ultra Wideband}
\newacronym{ieee}{IEEE}{institute of electrical and electronics engineers}
\newacronym{dB}{dB}{decibel}
\newacronym{cdf}{cdf}{cumulative density function}
\newacronym{ccdf}{ccdf}{complementary cumulative density function}
\newacronym{min}{Min.}{minimum}
\newacronym{med}{Med.}{median}
\newacronym{avg}{Avg.}{average}
\newacronym{ul}{UL}{up-link}
\newacronym{dl}{DL}{down-link}
\newacronym{app}{APP}{a-posteriori probability}
\newacronym{logmap}{LogMAP}{log maximum a-posteriori}
\newacronym{llr}{LLR}{log-likelihood ratio}
\newacronym{ue}{UE}{user equipment}
\newacronym{ai}{AI}{artificial intelligence}
\newacronym{qos}{QoS}{quality of service}
\newacronym{6g}{6G}{sixth generation}
\newacronym{5g}{5G}{fifth generation}
\newacronym{4g}{4G}{fourth generation}
\newacronym{tti}{TTI}{transmission time interval}
\newacronym{rrm}{RRM}{radio resource management}
\newacronym{mmib}{MMIB}{mean mutual information per bit}
\newacronym{dsi}{DSI}{decoder state information}
\newacronym{tb}{TB}{transport block}
\newacronym{isi}{ISI}{inter-symbol interference}
\newacronym{tbs}{TBS}{transport block size}
\newacronym{cb}{CB}{code block}
\newacronym{cbg}{CBG}{code block group}
\newacronym{cbs}{CBS}{code block size}
\newacronym{prb}{PRB}{physical resource block}
\newacronym{rb}{RB}{resource block}
\newacronym{bler}{BLER}{block error rate}
\newacronym{blep}{BLEP}{block error probability}
\newacronym{crc}{CRC}{cyclic redundancy check}
\newacronym{tdd}{TDD}{time division duplex}
\newacronym{fdd}{FDD}{frequency division duplex}
\newacronym{embb}{eMBB}{enhanced mobile broadband}
\newacronym{mbb}{MBB}{mobile broadband}
\newacronym{mcc}{MCC}{mission critical communication}
\newacronym{mmc}{MMC}{massive machine communication}
\newacronym{mtc}{MTC}{machine type of communication}
\newacronym{mmtc}{mMTC}{massive machine type of communication}
\newacronym{umtc}{uMTC}{ultra-reliable \gls{mtc}}
\newacronym{urllc}{URLLC}{ultra-reliable low-latency communications}
\newacronym{rtt}{RTT}{round trip time}
\newacronym{rs}{RS}{reference symbols}
\newacronym{kpi}{KPI}{key performance indicator}
\newacronym{kpis}{KPIs}{key performance indicators}
\newacronym{tx}{Tx}{transmitter node}
\newacronym{rx}{Rx}{receiver node}
\newacronym{cran}{C-RAN}{centralized radio access network}
\newacronym{rru}{RRU}{remote radio unit}
\newacronym{bbu}{BBU}{baseband unit}
\newacronym{fhd}{FHD}{fronthaul delay}
\newacronym{cch}{CCH}{control channel}
\newacronym{saw}{SAW}{stop-and-wait}
\newacronym{qci}{QCI}{\gls{qos} class identifier}
\newacronym{gbr}{GBR}{guaranteed bit rate}
\newacronym{mbr}{MBR}{maximum bit rate}
\newacronym{ngbr}{non-GBR}{non-\gls{gbr}}
\newacronym{arp}{ARP}{allocation and retention priority}
\newacronym{effcr}{ECR}{effective coding rate}
\newacronym{mcs}{MCS}{modulation and coding scheme}
\newacronym{eva}{EVA}{extended vehicular A}
\newacronym{epa}{EPA}{extended pedestrian A}
\newacronym{etu}{ETU}{extended typical urban}
\newacronym{re}{RE}{resource element}
\newacronym{reS}{REs}{resource elements}
\newacronym{nr}{NR}{new radio}
\newacronym{qpsk}{QPSK}{quadrature phase shift keying}
\newacronym{qam}{QAM}{quadrature amplitude modulation}
\newacronym{siso}{SISO}{single-input and single-output}	
\newacronym{miso}{MISO}{multiple-input single-output}
\newacronym{mimo}{MIMO}{multiple-input multiple-output}
\newacronym{bs}{BS}{base station}
\newacronym{phy}{PHY}{physical layer}
\newacronym{rlc}{RLC}{radio link control}
\newacronym{bcfsaw}{BCF-SAW}{BCF-SAW}
\newacronym{bcf}{BCF}{backwards composite feedback}
\newacronym{bac}{BAC}{binary asymmetric channel}
\newacronym{bsc}{BSC}{binary symmetric channel}
\newacronym{dtx}{DTX}{discontinued transmission}
\newacronym{bpsk}{BPSK}{binary phase shift keying}
\newacronym{bep}{BEP}{bit error probability}
\newacronym{ndi}{NDI}{new data indicator}
\newacronym{dci}{DCI}{downlink control information}
\newacronym{csit}{CSIT}{channel state information at the transmitter}
\newacronym{lt}{LT}{loudest talker}
\newacronym{ct}{CT}{cooperative transmission}
\newacronym{bps}{bps}{bits per second}
\newacronym{bpcu}{bpcu}{bits per channel use}
\newacronym{los}{LOS}{line-of-sight}
\newacronym{nlos}{NLOS}{non-line-of-sight}
\newacronym{regsaw}{Reg-SAW}{Regular SAW}
\newacronym{Lrep}{$L$-Rep-ACK}{Increased feedback repetition order}
\newacronym{Lack}{$L$-ACK-SAW}{$L$ required ACK per packet}
\newacronym{Asym}{Asym-SAW}{Asymmetric feedback detection for SAW}
\newacronym{bretx}{Blind-ReTx}{Blind retransmission}
\newacronym{af}{AF}{amplify-and-forward}
\newacronym{ap}{AP}{access point}
\newacronym{icn}{ICN}{industrial control network}
\newacronym{comp}{CoMP}{coordinated multi-point}
\newacronym{rhs}{RHS}{right hand side}
\newacronym{lhs}{LHS}{left hand side}
\newacronym{sumiso}{SU-MISO}{single-user multiple-input-single-output}
\newacronym{ibl}{IBL}{infinite block length}
\newacronym{fbl}{FBL}{finite block length}
\newacronym{lan}{LAN}{local area network}
\newacronym{wsn}{WSN}{wireless sensor network}
\newacronym{rt}{RT}{real-time}
\newacronym{tdm}{TDM}{time division multeplxing}
\newacronym{nist}{NIST}{National Institute of Standards and Technology}
\newacronym{cbrs}{CBRS}{Citizens Broadband Radio Service}
\newacronym{itu}{ITU}{International Telecommunications Union}
\newacronym{mmwave}{mmWave}{millimeter-wave}
\newacronym{nsr}{NSR}{noise-to-signal ratio}
\newacronym{das}{DAS}{distributed antenna system}
\newacronym{pmf}{PMF}{probability mass function}
\newacronym{srs}{SRS}{sounding reference signal}
\newacronym{dmrs}{DMRS}{demodulation reference signal}
\newacronym{psd}{PSD}{power spectral density}
\newacronym{rf}{RF}{radio frequency}
\newacronym{id}{ID}{identifier}
\newacronym{df}{DF}{decode-and-forward}
\newacronym{iot}{IoT}{internet of things}
\newacronym{iiot}{IIoT}{industrial internet-of-things}
\newacronym{iab}{I-AB}{integrated access and backhaul}
\newacronym{icsi}{I-CSI}{imperfect CSI}
\newacronym{pcsi}{P-CSI}{perfect CSI}
\newacronym{andcoop}{ANDCoop}{adaptive network-device cooperation}
\newacronym{dmt}{DMT}{diversity-multiplexing tradeoff}
\newacronym{scs}{SCS}{sub-carrier spacing}
\newacronym{trp}{TRP}{transmission reception point}
\newacronym{e2e}{E2E}{end-to-end}
\newacronym{d2d}{D2D}{device-to-device}
\newacronym{prose}{ProSe}{proximity service}
\newacronym{im}{IM}{interference management}
\newacronym{aoi}{AoI}{age-of-information}
\newacronym{v2v}{V2V}{vehicle-to-vehicle}
\newacronym{cdd}{CDD}{cyclic delay diversity}
\newacronym{ldpc}{LDPC}{low-density parity-check }
\newacronym{InF}{InF}{indoor factory}
\newacronym{scm}{SCM}{spatial channel modeling}
\newacronym{cu}{CU}{central unit}
\newacronym{du}{DU}{distributed unit}
\newacronym{sl}{SL}{side-link}
\newacronym{ris}{RIS}{reflective intelligent surface}
\newacronym{dvbt}{DVB-T}{digital video broadcasting-terrestrial}
\newacronym{sa}{SA}{service availability}
\newacronym{sfn}{SFN}{single-frequency network}
\newacronym{twc}{TWC}{tactile wireless control}
\newacronym{ncr}{NCR}{network-controlled repeater}
\newacronym{map}{MAP}{maximum a posteriori}
\newacronym{jscc}{JSCC}{joint source-channel coding}
\newacronym{vqvae}{VQ-VAE}{Vector Quantized-Variational Autoencoder}

\begin{document}
\title{Generative Decompression: Optimal Lossy Decoding Against Distribution Mismatch}

\author{Saeed~R.~Khosravirad,~\IEEEmembership{Senior Member,~IEEE,}
        Ahmed Alkhateeb,~\IEEEmembership{Fellow,~IEEE,}
        and~Ingrid~van~de~Voorde,~\IEEEmembership{Member,~IEEE}
\thanks{Saeed R. Khosravirad is with the Nokia Bell Laboratories, Murray Hill, NJ 07974 USA e-mail: (saeed.khosravirad@nokia-bell-labs.com).}
\thanks{Ahmed Alkhateeb is with the Wireless Intelligence Laboratory and the School of Electrical, Computer and Energy Engineering, Arizona State University, Tempe, AZ 85281 USA (e-mail:alkhateeb@asu.edu).}
\thanks{Ingrid van de Voorde is with the Nokia Bell Laboratories, Antwerpen, Belgium (email:ingrid.van\_de\_voorde@nokia-bell-labs.com)}}


\maketitle
\begin{abstract}
This paper addresses optimal decoding strategies in lossy compression where the assumed distribution for compressor design mismatches the actual (true) distribution of the source. This problem has immediate relevance in standardized communication systems where the decoder acquires side information or priors about the true distribution that are unavailable to the fixed encoder. We formally define the \emph{mismatched quantization problem}, demonstrating that the optimal reconstruction rule, termed \emph{generative decompression}, aligns with classical Bayesian estimation by taking the conditional expectation under the true distribution given the quantization indices and adapting it to fixed-encoder constraints. This strategy effectively performs a generative Bayesian correction on the decoder side, strictly outperforming the conventional centroid rule. We extend this framework to transmission over noisy channels, deriving a robust soft-decoding rule that quantifies the inefficiency of standard modular source--channel separation architectures under mismatch. Furthermore, we generalize the approach to task-oriented decoding, showing that the optimal strategy shifts from conditional mean estimation to maximum a posteriori (MAP) detection. Experimental results on Gaussian sources and deep-learning-based semantic classification demonstrate that generative decompression closes a vast majority of the performance gap to the ideal joint-optimization benchmark, enabling adaptive, high-fidelity reconstruction without modifying the encoder.
\end{abstract}


\section{Introduction}
Lossy data compression is a foundational building block in communication and storage of media and telemetry. In a nutshell, continuous-valued random variables are mapped to a finite set of discrete levels to enable efficient representation \cite{GrayNeuhoff1998}. Typically, a compressor is designed to minimize distortion, e.g., \gls{mse}, under an assumed source distribution. However, in practical scenarios, the \emph{true} source distribution may deviate from the design assumption due to estimation errors, non-stationarity, or model inaccuracies \cite{gray2003mismatch}.

When such distribution mismatch is known to the receiver, the decoder can exploit this side information to adjust its reconstruction rule and reduce distortion. This concept is rooted in classical information theory, specifically the Wyner-Ziv theorem on source coding with side information at the decoder~\cite{1055508}, but it has gained renewed relevance in more recent research and development efforts regarding generative and semantic communication.

\subsection{Motivating Scenarios}
\paragraph{Context-Aware Semantic Decoding}
In semantic compression, the decoder's goal may shift over time. While the message was encoded for signal fidelity, the receiver may be interested in specific semantic features relevant to a downstream task (e.g., classification) \cite{GulerYener2014}. A task-aware decoder can adjust reconstructions to minimize semantic error~\cite{8493595}.

\paragraph{CSI Compression}
In wireless communications, \gls{csi} is often compressed by a standardized encoder assuming a generic channel model. The base station receiver, however, learns the specific channel statistics of the user over time. By updating the decoder to reflect these learned statistics, the reconstruction error of the CSI feedback can be significantly reduced~\cite{luo2025generative}.

\paragraph{Federated Learning and Broadcasting}
Consider a broadcasting with receiver side-information scenario. A central server (Encoder) broadcasts a model or signal compressed using a global codebook optimized for the average user statistics. However, each receiver possesses distinct, local statistical characteristics. Each decoder can leverage its local \emph{true} distribution to reinterpret the received quantization indices, effectively personalizing the decompression without altering the bitstream.

\subsection{Literature Review}

The problem of optimal coding under statistical mismatch has roots in classical information theory. Lapidoth \cite{lapidoth1997role} characterized the mismatch capacity of channel decoding, while Gray and Linder \cite{gray2003mismatch} extended this to high-rate vector quantization, proving that the penalty for designing a quantizer on a mismatched distribution is determined by the relative entropy (Kullback-Leibler divergence) between the design and true distributions.

This framework relates closely to source coding with side information. The seminal Slepian-Wolf theorem \cite{slepian1973noiseless} and its lossy extension by Wyner and Ziv \cite{1055508} established that for correlated sources, side information available only at the decoder allows for compression efficiency approaching that of the informed-encoder case. This theoretical possibility of shifting complexity to the receiver is the basis for modern distributed source coding strategies \cite{xiong2004distributed, pradhan2003distributed}.

In the deep learning era, handling distribution mismatch has shifted from analytical point-density optimizations to data-driven generative modeling. Blau and Michaeli \cite{blau2018perception} formally defined the perception-distortion tradeoff, proving that minimizing distortion alone leads to statistically unlikely reconstructions. To traverse this optimal frontier, Tschannen et al. \cite{tschannen2018deep} proposed distribution-preserving lossy compression (DPLC), where the decoder is constrained to output samples from the true data manifold. This generative approach is further supported by adversarial training frameworks \cite{agustsson2019generative, mentzer2020high}, which ensure perceptual high fidelity at low bitrates. These methods align with advances in neural discrete representation learning \cite{vandenoord2017neural} and end-to-end optimized compression \cite{balle2017end}, which learn robust codes that can generalize across distributional shifts.

These principles are increasingly applied to semantic and task-oriented communications \cite{gunduz2022beyond, qin2021semantic, strinati20216g}. Here, the decoder optimizes for a downstream task utility rather than signal fidelity, often mitigating the impact of channel or source mismatch through semantic resilience. A critical application of this paradigm is CSI feedback in massive \gls{mimo}. While deep learning-based compression like CsiNet \cite{wen2018deep} has shown promise, recent work \cite{luo2025generative} demonstrates that generative decoding at the base station—leveraging site-specific digital twin data as side information—can correct for the quantization artifacts of standardized user equipment encoders. Finally, finite-blocklength analyses by Zhou et al. \cite{zhou2019dispersion} and recent Deep Joint Source-Channel Coding (Deep JSCC) advances \cite{bourtsoulatze2019deep} illustrate that under distribution mismatch and finite blocklength, strictly modular source–channel architectures can incur a non-negligible performance penalty, motivating joint source–channel decoding strategies even though Shannon's separation theorem remains asymptotically optimal under its standard joint-design assumptions.

\subsection{Contributions}
Despite these advances, a critical gap remains in applying these principles to systems with rigid, standardized encoders. Most neural compression frameworks require joint training of encoder and decoder, rendering them incompatible with legacy hardware or fixed protocols like \gls{5g} CSI feedback. Conversely, while classical mismatch theory describes the fundamental limits of sub-optimal coding, it offers little prescription for practical, high-dimensional reconstruction under severe shifts. This paper addresses this by formulating decoding as a Bayesian inference problem where the encoder is immutable. We propose \emph{generative decompression} to demonstrate that a purely decoder-side adjustment—leveraging generative priors of the true distribution—can recover significant performance lost to distribution mismatch without requiring modifications to the transmitter. Throughout this work, we use the term ``generative decompression'' to denote Bayesian reconstruction using a generative prior over the source, rather than stochastic sample generation as in GAN or diffusion models.
The main contributions of this paper are as follows:

\begin{itemize}
    \item We formalize the \emph{mismatched quantization problem} where the quantizer is fixed to a design distribution $P_X^{(d)}$, but the decoder knows the true distribution $P_X^{(t)}$. We show that the optimal reconstruction rule—termed \emph{generative decompression}—aligns with the classical \gls{mmse} estimator, specifically the conditional expectation under the true distribution. This simple Bayesian correction strictly outperforms the conventional centroid rule whenever mismatch is present. While the resulting estimator is mathematically equivalent to the classical \gls{mmse} solution, the novelty lies in applying it as a strictly decoder-side patch to standardized, immutable quantization protocols.

    \item We extend the framework to noisy channels and derive the optimal soft-decoding rule, which inseparably couples the true source priors with channel reliability statistics. We quantify the excess distortion incurred by standard modular source--channel separation architectures with hard-decision interfaces under distribution mismatch and show how a joint soft-decoding rule removes this penalty.

    \item We provide closed-form expressions and high-rate asymptotic theory for Gaussian and Laplace mismatch scenarios. We identify two distinct regimes—central mismatch (modest, vanishing gains) and tail mismatch (large, rate-persistent gains)—and prove that generative decompression removes distortion floors caused by overload bias, recovering the ideal $O(2^{-2b})$ decay in many practical cases.

    \item We show that the same conditional-expectation principle extends to arbitrary task losses, yielding task-optimal reconstructions instead of \gls{mse}. We demonstrate the framework on \gls{5g} \gls{csi} feedback with a weighted-MSE beamforming loss and on deep \gls{vqvae} models for semantic classification, showing significant decoder-only gains in both analytic and high-dimensional learned settings.
\end{itemize}

\subsection{Paper Organization}
The rest of the paper is organized as follows.
Section~\ref{sec:problem} formalizes the mismatched quantization problem, introduces notation, and details the generative decompression optimization for the noiseless case.
Section~\ref{sec:scalar} specializes the theory to scalar quantization, derives explicit distortion expressions for Gaussian and Laplace sources, and presents high-rate asymptotic analysis together with numerical validation.
Section~\ref{sec:noisy} extends the framework to transmission over noisy channels, and rigorously quantifies the inefficiency of modular source–channel separation under mismatch when the encoder is fixed.
Section~\ref{sec:taskaware} generalizes the approach to arbitrary task-specific distortion measures, with analytic insights for CSI beamforming feedback and large-scale experimental validation for semantic classification under class and domain shifts.
Finally, Section~\ref{sec:conclusion} concludes the paper and discusses broader implications.

\section{Problem Formulation}
\label{sec:problem}

\subsection{Source Models and Notation}
Let $X$ be a real-valued random variable. The \emph{true} law of $X$ is denoted by $P_X^{(t)}$, with PDF $f_t(x)$. The encoder designs the quantizer under a possibly mismatched \emph{design} law $P_X^{(d)}$ with PDF $f_d(x)$. Throughout, we use the shorthand:
\[
\mathbb{E}^{(t)}[\cdot] \triangleq \mathbb{E}_{X \sim P_X^{(t)}}[\cdot], \quad
\mathbb{P}^{(t)}(\cdot) \triangleq \mathbb{P}_{X \sim P_X^{(t)}}(\cdot),
\]
and similarly for the design law, and use $\mathbb{P}$ as probability measure.

In practice, the decoder does not have analytic access to $P_X^{(t)}$ but can infer it from long-term observations, side information, or site-specific digital twins. We therefore view $P_X^{(t)}$ as a statistical model available at the decoder but not at the fixed encoder, which is designed once under $P_X^{(d)}$.

\subsection{Quantizer Designed Under the Mismatched Model}
The encoder designs a $b$-bit scalar quantizer under $P_X^{(d)}$. Let $N = 2^b$, and let $Q_d : \mathbb{R} \to \{1,\dots,N\}$ be the quantizer, inducing partition $\{\mathcal{R}_i\}_{i=1}^N$.
The optimal reconstruction values under $P_X^{(d)}$ are the centroids:
\begin{align}
a_i^d &= \mathbb{E}_{d}[X \mid X \in \mathcal{R}_i].
\end{align}

\subsection{Generative Decompression Optimization}
The source follows $f_t(x)$, yielding index $I = Q_d(X)$. The decoder seeks a mapping $g_t(i) = a_i^t$ to minimize the expected distortion under the true distribution, given the fixed partition boundaries determined by the encoder:
\begin{align}
D_t(\mathbf{a}^t) &\triangleq \mathbb{E}_{t}\big[(X - a_I^t)^2\big].
\end{align}

\begin{proposition}
\label{prop:optimal-recon-true}
The reconstruction points that minimize $D_t(\mathbf{a}^t)$ correspond to the MMSE estimator under the true distribution:
\begin{align}
a_i^{t\star} &= \mathbb{E}_{t}[X \mid X \in \mathcal{R}_i].
\end{align}
The minimal distortion is
\begin{align}
D_{\min} &= \sum_{i=1}^N \mathbb{P}_{t}(I = i) \, \mathrm{Var}_{t}[X \mid I = i].
\end{align}
\end{proposition}

\begin{IEEEproof}
This result is a direct application of the orthogonality principle of minimum mean-square estimation. For any squared-error cost function, the optimal estimator is the conditional expectation of the random variable given the observation. Here, the observation is the event $\{X \in \mathcal{R}_i\}$. Applying the expectation conditionally to each quantization bin yields the result.
\end{IEEEproof}

The optimal reconstruction rule is effectively the classical Lloyd-Max centroid condition applied under the true distribution rather than the design one. The excess distortion is $\Delta D \triangleq D_t(\mathbf{a}^d) - D_{\min} \ge 0$.

\section{Case Study: Scalar Quantization}
\label{sec:scalar}

\subsection{One-Bit Quantization (Gaussian True Source)}

We specialize the formulation to Gaussian sources. Let the design law be $X \sim \mathcal{N}(\mu_0,\sigma_0^2)$ and the true law be $X \sim \mathcal{N}(\mu_1,\sigma_1^2)$.

The optimal 1-bit quantizer for the design law has a threshold at the mean $\tau = \mu_0$ and reconstruction points:
\[
a_{1,2}^d = \mu_0 \mp \sigma_0 \sqrt{\frac{2}{\pi}}.
\]

\subsubsection{Unadjusted Decoder Distortion}
When this quantizer is applied to the true source $X \sim \mathcal{N}(\mu_1,\sigma_1^2)$, the normalized threshold is $\alpha \triangleq (\mu_0 - \mu_1)/\sigma_1$. The distortion using the unadjusted design centroids is:
\begin{equation}
D_t(\mathbf{a}^d) = \sigma_1^2 \mathbb{E}_{Z}\left[ (Z - Q'(Z))^2 \right],
\end{equation}
where $Q'(Z)$ maps the standardized variable $Z$ to normalized design centroids. Using truncated normal moments, this expands to:
\begin{align}
D_t(\mathbf{a}^d) &= \sigma_1^2 \Big( \Phi(\alpha)[1 - \alpha\lambda_L + 2a_1'\lambda_L + (a_1')^2] \nonumber \\
&\quad + (1-\Phi(\alpha))[1 + \alpha\lambda_R - 2a_2'\lambda_R + (a_2')^2] \Big),
\end{align}
where  $\phi(\cdot)$ and $\Phi(\cdot)$ denote the standard normal PDF and CDF, respectively, and $\lambda_L = \frac{\phi(\alpha)}{\Phi(\alpha)}$ and $\lambda_R = \frac{\phi(\alpha)}{1-\Phi(\alpha)}$ are the Inverse Mills Ratios.

\subsubsection{Optimal Decoder Adjustment}
Using Prop. \ref{prop:optimal-recon-true}, the optimal decoder adjusts the reconstruction points to the true conditional means:
\begin{align}
a_1^{t\star} &= \mathbb{E}[X \mid X < \mu_0] = \mu_1 - \sigma_1 \lambda_L, \\
a_2^{t\star} &= \mathbb{E}[X \mid X \ge \mu_0] = \mu_1 + \sigma_1 \lambda_R.
\end{align}
The minimal achievable distortion $D_{\min}$ is given by the weighted sum of conditional variances:
\begin{align}\nonumber
D_{\min} =  \sigma_1^2 \Big(  & \Phi(\alpha)[1 - \alpha\lambda_L - \lambda_L^2] \\ 
& + (1-\Phi(\alpha))[1 + \alpha\lambda_R - \lambda_R^2] \Big).
\end{align}

\subsubsection{Case of Variance Mismatch}\label{sec:2overpi}
Consider the case where means are matched ($\mu_0=\mu_1=0$) but variances differ ($\sigma_1 \neq \sigma_0$). Here $\alpha=0$ and $\lambda_L=\lambda_R=\sqrt{2/\pi}$.
The unadjusted distortion is $D_t(\mathbf{a}^d) = \sigma_1^2(1 + \frac{2}{\pi\sigma_1^2})$.
The optimal distortion is $D_{\min} = \sigma_1^2(1 - \frac{2}{\pi})$.
The relative gain from decoder adaptation is:
\begin{equation}
\text{Gain} = 1 - \frac{D_{\min}}{D_t(\mathbf{a}^d)} = 1 - \frac{1 - 2/\pi}{1 + 2/(\pi\sigma_1^2)}.
\end{equation}
As $\sigma_1 \to \infty$ (severe under-estimation of variance by the encoder), the gain approaches $2/\pi \approx 64\%$. This  result proves that simple decoder-side rescaling can recover a majority of the coding loss in high-variance mismatch scenarios. Numerical experiments presented later in this section 
demonstrate that the $2/\pi$ limit is applicable to the case of Gaussian variance mismatch regardless of the number of quantization bins.

\subsection{Multi-Bit Quantization and High-Rate Analysis}
For $b > 1$ bits, the quantization partitions the real line into bounded inner cells and semi-infinite outer cells (overload regions). Let the encoder's support region be $S_d = [\tau_0, \tau_{N-1}]$ and the overload region be $\mathcal{O}_d = \mathbb{R} \setminus S_d$.
To analyze the distortion for large $N$, we decompose the total distortion into granular (inner) and overload (outer) components:
\begin{equation}
    D \approx \underbrace{\frac{1}{12 N^2} \int_{S_d} \frac{f_t(x)}{\lambda_d(x)^2} dx}_{D_{\text{granular}}} + \underbrace{\sum_{i \in \{1, N\}} \int_{\mathcal{R}_i \cap \mathcal{O}_d} (x - a_i)^2 f_t(x) dx}_{D_{\text{overload}}}.
\end{equation}
The granular term $D_{\text{granular}}$ is the Bennett integral~\cite{391237} where $\lambda_d(x)$ denotes the point density of the fixed encoder.  $D_{\text{granular}}$ decays as $O(N^{-2})$ regardless of the decoding rule, provided the density function is smooth. The asymptotic performance is therefore governed by $D_{\text{overload}}$.
Let us analyze the outermost bin $\mathcal{R}_N = (\tau_{N-1}, \infty)$. Let $\mu_{\text{tail}}^{(t)} = \mathbb{E}_t[X | X > \tau_{N-1}]$ be the conditional mean of the tail under the true distribution.
Let $p_i^{(t)} \triangleq \mathbb{P}_{t}(X \in \mathcal{R}_i)$ denote the true probability that the source $X$ falls into the $i$-th bin $\mathcal{R}_i$.
Using the bias-variance decomposition of MSE~\cite{räisä2025biasvariancedecompositionensemblesmultiple}, the distortion in this bin for an arbitrary reconstruction value $a_N$ is:
\begin{align}
    D_{\text{overload}}^{(N)} &= \mathbb{P}_t(\mathcal{R}_N) \cdot \mathbb{E}_t [ (X - a_N)^2 \mid X \in \mathcal{R}_N ] \nonumber \\
    &= p_N^{(t)} \left( \mathrm{Var}_t[X \mid X \in \mathcal{R}_N] 
    + (\mu_{\text{tail}}^{(t)} - a_N)^2\right).
\end{align}
The variance term is irreducible while the squared bias term $(\mu_{\text{tail}}^{(t)} - a_N)^2$ depends on the decoder.
\subsubsection{Unadjusted Decoder}
The fixed encoder sets $a_N = a_N^d$, the centroid for the \emph{design} distribution. If $f_t$ has heavier tails than $f_d$ (e.g., Laplace vs. Gaussian, or $\sigma_1 \gg \sigma_0$), the true mass shifts significantly outward such that $\mu_{\text{tail}}^{(t)} \gg a_N^d$. This results in a large, strictly positive bias term $(\mu_{\text{tail}}^{(t)} - a_N^d)^2$ which dominates the distortion as the granular error vanishes with increasing rate.
\subsubsection{Generative Decoder}
The proposed decoder sets $a_N = a_N^{t\star} = \mu_{\text{tail}}^{(t)}$. This choice sets the bias term to zero. The remaining distortion for the generative decoder is strictly the conditional variance.
\begin{equation}
    D_{\text{overload}} = p_N^{(t)} \mathrm{Var}_t[X \mid X > \tau_{N-1}].
\end{equation}
This leads to two distinct asymptotic regimes for the generative decompression:
\begin{itemize}
    \item \textbf{Central Mismatch:} If tails are matched (i.e., probability mass in the tail decays at the same rate), the bias term in the unadjusted distortion is negligible. The dominant error is $D_{\text{granular}}$, so the relative gain of adaptation diminishes as $O(N^{-2})$ for large $b$.
    \item \textbf{Tail Mismatch:} If $f_t$ is heavy-tailed relative to $f_d$, the bias term from the unadjusted decoder does not vanish (or decays slower than $N^{-2}$), effectively creating an error floor. Generative decompression removes this floor, allowing the decoder to virtually extend the dynamic range of the quantizer, yielding gains that persist or grow with $N$.
\end{itemize}

\begin{figure}[t]
    \centering
    \includegraphics[width=0.48\textwidth]{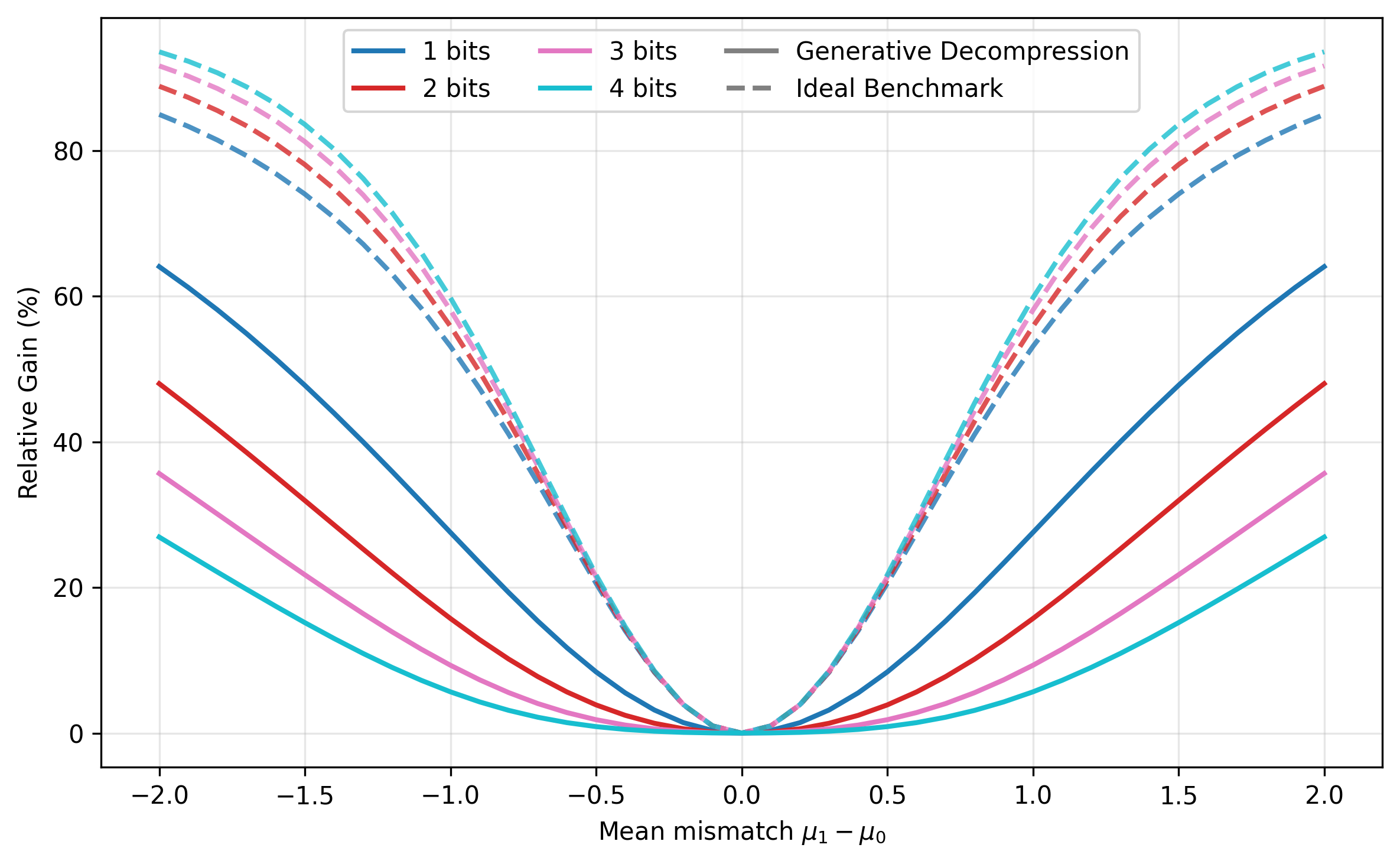}
    \caption{Decoder-side gain vs.\ mean mismatch for Gaussian sources. Quantizers designed for $\mathcal{N}(0,1)$, evaluated on $\mathcal{N}(\mu_1, 1)$.}
    \label{fig:exp1_mean}
\end{figure}

To quantify the distortion reduction predicted above, we performed experiments with a fixed encoder designed for the standard Gaussian $\mathcal{N}(0,1)$ using the Lloyd–Max algorithm (60 iterations, $3\times10^{5}$ training samples). We tested bit depths $b \in \{1,2,3,4\}$ ($N \in \{2,4,8,16\}$). Performance is reported as 
\begin{equation}\label{eq:relgain}
    \text{Relative Gain} = \bigl(1 - \frac{D_{\min}}{D_{\text{std}}}\bigr) \times 100,
\end{equation}
where $D_{\text{std}}$ uses the original design centroids and $D_{\min}$ uses the optimal true conditional means.

\paragraph{Mean drift}
Figure~\ref{fig:exp1_mean} shows the relative gain when the true source is $\mathcal{N}(\mu_1,1)$, $\mu_1 \in [-2,2]$.
The gain is negligible near $\mu_1=0$ but rises symmetrically to nearly 60\% at $|\mu_1|=2$. This confirms that decoder adaptation effectively re-centers the reconstruction when the source drifts. The ideal benchmark case is also demonstrated (see Sec.~\ref{sec:idealbenchmark}).  

\paragraph{Variance scaling}
Figure~\ref{fig:exp1_var} shows the gain versus variance ratio $r = \sigma_1/\sigma_0$ (log scale).
\begin{figure}[t]
    \centering
    \includegraphics[width=0.48\textwidth]{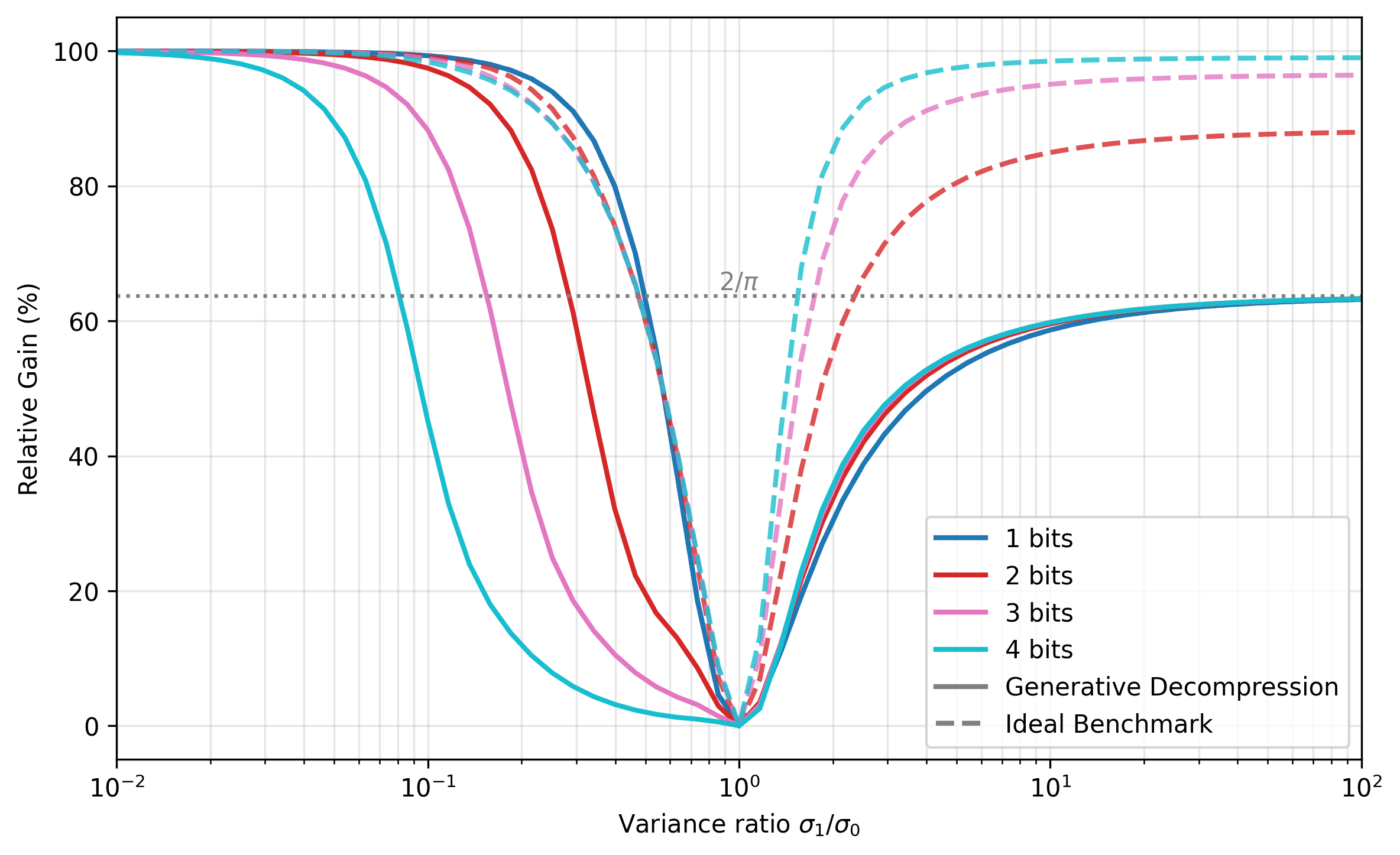}
    \caption{Decoder-side gain vs.\ variance ratio for Gaussian sources. Quantizers designed for $\mathcal{N}(0,1)$, evaluated on $\mathcal{N}(0, \sigma_1^2)$.}
    \label{fig:exp1_var}
\end{figure}
The characteristic U-shape appears: in the over-dispersed regime ($\sigma_1 \gg \sigma_0$), higher bit depths yield larger gains because outer-bin centroids can be shifted far outward, effectively un-clipping the tails—consistent with the tail-mismatch regime derived above. In the under-dispersed regime, the decoder shrinks reconstruction levels toward zero to reduce unnecessary granular noise. For 1 bit case, the ideal benchmark and the generative decoder performances are equivalent.

\begin{table}[h!]
\centering
\caption{Generative decoding for Laplace true distribution.}
\label{tab:laplace}
\begin{tabular}{ccccc}
\toprule
Bits & 1 & 2 & 3 & 4 \\
\midrule
Generative decoder relative gain & 1.65\% & 6.09\% & 11.01\% & 16.74\% \\
Ideal benchmark relative gain  & 1.65\% & 7.84\% & 24.78\% & 41.87\% \\
\bottomrule
\end{tabular}
\end{table}

\paragraph{Laplace distribution}
The Laplace distribution has tails that decay as $\propto e^{-|x|}$, while the Gaussian distribution has tails that decay more rapidly as $\propto e^{-{x^2}}$. On a zero-mean unit-variance Laplace source, gains from generative decompression increase monotonically with rate as shown in Table~\ref{tab:laplace}. The heavier Laplacian tails are captured in dedicated outer bins at higher rates, where the generative decoder can aggressively correct the severe under-estimation of tail mass by the Gaussian design—exactly the tail-mismatch scenario that produces persistent high-rate gains.



These experiments numerically confirm the theoretical predictions: decoder adaptation yields modest gains under pure central mismatch but dramatic, rate-persistent (or even increasing) gains whenever the true distribution places more mass in the tails than the design distribution anticipated.

\subsection{Benchmarking Against Ideal Matched Quantization}\label{sec:idealbenchmark}

To rigorously assess the efficacy of generative decompression, we must compare it not only against the  unadjusted decoder but against the \emph{ideal benchmark}: a system where both encoder and decoder are perfectly matched to the true distribution $f_t(x)$.

Let the point density function of the fixed encoder be $\lambda_d(x) \propto f_d(x)^{1/3}$. We define three distortion values:
\begin{enumerate}
 \item {Ideal Benchmark ($D_{\text{ideal}}$):} The encoder and decoder are both retrained for $f_t(x)$. By the Panter-Dite asymptotic relation~\cite{GrayNeuhoff1998,1701410}:
\begin{align}
 D_{\text{ideal}} \approx \frac{1}{12 N^2} \|f_t\|_{1/3}^3, 
\end{align}
where $\|f\|_\alpha \triangleq \left(\int f^\alpha dx\right)^{1/\alpha}$.
 \item {Unadjusted Decoder ($D_{\text{fix}}$):} The decoder uses centroids designed for $f_d$. This introduces a systematic bias in every bin, which does not vanish as $N \to \infty$ if the support of $f_t$ exceeds $f_d$.
\item {Generative Decoder ($D_{\text{gen}}$):} The encoder is fixed to $\lambda_d$, but reconstruction is optimal. The distortion is given by the mismatched Bennett integral~\cite{391237}:
 \begin{equation}
D_{\text{gen}} \approx \frac{1}{12 N^2} \int \frac{f_t(x)}{\lambda_d(x)^2} dx.
\end{equation}
\end{enumerate}

\begin{definition}[Mismatch penalty factor]
Let $\mathcal{L}$ be the asymptotic ratio of the generative distortion to the ideal distortion:
\begin{equation}\label{eq:asymptoticratio}
\mathcal{L} \triangleq \lim_{N \to \infty} \frac{D_{\text{gen}}}{D_{\text{ideal}}} = \frac{\int f_t(x) \lambda_d(x)^{-2} dx}{\|f_t\|_{1/3}^3}.
\end{equation}
\end{definition}
This factor allows us to  differentiate between recoverable and irrecoverable regimes for the generative decoder. The integral in \eqref{eq:asymptoticratio} is finite under mild regularity conditions—for example, when the design density does not vanish faster than the true density in the tails—so that the mismatch penalty factor is well-defined.

We first examine the fundamental limit of binary quantization ($N=1$). This regime admits a strong optimality result for scale-family distributions (e.g., Gaussian or Laplace with varying variance).

\begin{theorem}[Perfect optimality at $N=1$]
\label{thm:perfect_N1}
Let the design and true distributions be symmetric and unimodal about the same mean, differing only in variance (scale). For a 1-bit scalar quantizer, the generative decoder achieves the ideal distortion exactly:
\begin{equation}
    D_{\text{gen}} = D_{\text{ideal}} < D_{\text{fix}}.
\end{equation}
In this regime, the Mismatch Penalty Factor vanishes ($\mathcal{L} = 1$). This statement holds within the considered scalar, 1-bit setting.
\end{theorem}

\begin{IEEEproof}
For any symmetric unimodal distribution centered at $\mu$, the optimal 1-bit quantization threshold is unique and equal to $\mu$. Since the design and true distributions share $\mu$, the fixed encoder's partition $\mathcal{R}_1 = (-\infty, \mu), \mathcal{R}_2 = [\mu, \infty)$ is identical to the partition of the ideal encoder optimized for the true source.
Since the partitions are identical, $D_{\text{gen}} = D_{\text{ideal}}$.
The unadjusted decoder, however, uses centroids $a^d \neq a^{t\star}$, yielding strictly higher distortion due to the squared bias term derived in Sec.~\ref{sec:scalar}-A.
\end{IEEEproof}

\begin{theorem}[Distortion hierarchy]
\label{thm:hierarchy}
For any blocklength $N$ and any mismatch scenario where the design partition differs from the optimal true partition, the distortions satisfy the strict hierarchy:
\begin{equation}
    D_{\text{ideal}} < D_{\text{gen}} < D_{\text{fix}}.
\end{equation}
Equality $D_{\text{gen}} = D_{\text{fix}}$ holds only if the design centroids coincidentally match the true conditional means. Equality $D_{\text{ideal}} = D_{\text{gen}}$ holds only if the design partition is optimal for the true source (e.g., the symmetric $N=1$ case in Theorem \ref{thm:perfect_N1}).
\end{theorem}

\begin{IEEEproof}
The inequality $D_{\text{gen}} \le D_{\text{fix}}$ follows from the definition of the generative decoder in Prop. \ref{prop:optimal-recon-true}. For a fixed partition $\mathcal{P}$, the function $g(\mathbf{a}) = \mathbb{E}_t[(X - a_I)^2]$ is strictly convex and minimized uniquely by the conditional means $\mathbf{a}^{t\star}$. Since the unadjusted decoder uses centroids $\mathbf{a}^d \neq \mathbf{a}^{t\star}$ (under non-trivial mismatch), it holds strictly that $D_{\text{gen}} < D_{\text{fix}}$.

The inequality $D_{\text{ideal}} \le D_{\text{gen}}$ follows from the constraints on the optimization space. Let $\mathbb{Q}_N$ be the set of all possible $N$-level quantizers (pairs of partitions and codebooks). The ideal distortion is the global minimum: $D_{\text{ideal}} = \min_{(P, C) \in \mathbb{Q}_N} \mathbb{E}_t[(X - Q(X))^2]$.
The generative distortion is a constrained minimum where the partition $P$ is fixed to the design partition $P_d$: $D_{\text{gen}} = \min_{C} \mathbb{E}_t[(X - Q_{P_d, C}(X))^2]$.
Since the constrained minimum cannot be strictly lower than the global minimum, $D_{\text{ideal}} \le D_{\text{gen}}$. This is strict whenever $P_d$ is not the optimal Voronoi partition for $f_t$, which is true for $N>1$ under variance mismatch as the optimal cell widths depend on the source density.
\end{IEEEproof}

For higher rates ($N \to \infty$), the encoder's partition density becomes the dominant factor. Here, the fixed encoder is no longer optimal, but the generative decoder provides a crucial correction to the convergence rate.

\begin{theorem}[Asymptotic rate recovery]
\label{thm:rate_recovery}
Assume that the true distribution $P_X^{(t)}$ is absolutely continuous and has heavier tails than the design distribution $P_X^{(d)}$ (e.g., $f_t(x)/f_d(x)$ is eventually non-decreasing in $|x|$ and bounded away from zero on the overload region). In the presence of such distribution mismatch, the unadjusted decoder distortion saturates at a constant floor, while the generative decoder recovers the ideal asymptotic decay rate of $O(N^{-2})$.
\begin{align}
\lim_{N \to \infty} D_{\text{fix}} &> 0, \\
\lim_{N \to \infty} N^2 D_{\text{gen}} &= \text{constant}.
\end{align}
\end{theorem}
\begin{IEEEproof}
The unadjusted decoder error decomposes into variance (granular) + bias squared. As $N \to \infty$, granular error vanishes, but the bias $\mathbb{E}_t[(X - a_I^d)^2 \mid \mathcal{O}_d]$ remains strictly positive due to the mismatch in the overload regions (see Sec.~\ref{sec:scalar}-B). The generative decoder removes this bias by setting centroids to conditional means $a_I^{t\star}$, leaving only the conditional variance, which scales with the cell width squared $\Delta_i^2 \propto N^{-2}$.
\end{IEEEproof}

While Theorem \ref{thm:rate_recovery} proves that generative decompression restores the correct \emph{rate} of convergence, there is a constant penalty factor compared to the ideal benchmark.

\begin{lemma}[Asymmetry of Gaussian mismatch]
\label{lemma:asymmetry}
Let true and design distributions be zero-mean Gaussians with variances $\sigma_1^2$ and $\sigma_0^2$, respectively. Let $r = \sigma_1 / \sigma_0$ be the mismatch ratio. The penalty factor $\mathcal{L}$ has an asymmetric dependence on $r$:
\begin{itemize}
\item {Resolution loss, $r \to 0$:} If the source is narrower than designed ($\sigma_1 \ll \sigma_0$), the penalty diverges: $\mathcal{L} \to \infty$.
\item {Tail mismatch, $r > 1$:} If the source is wider than designed ($\sigma_1 > \sigma_0$), the penalty $\mathcal{L}$ grows slowly, provided the asymptotic integral converges.
\end{itemize}
\end{lemma}
\begin{IEEEproof}
Substituting the Gaussian point density $\lambda_d(x) \propto \exp(-x^2/3\sigma_0^2)$ into the definition of $\mathcal{L}$ in \eqref{eq:asymptoticratio}, the integrand's central term determines the granular error coefficient.
When $r \to 0$, the ratio $\lambda_d(x)^{-2}$ introduces excessive weight because the design density is spread far too thinly. This coarse granularity dominates the distortion, and no generative correction can recover the lost resolution, hence $\mathcal{L} \to \infty$.
When $r > 1$, the fixed encoder has finer resolution than needed for the source core. While this is inefficient, the penalty $\mathcal{L}$ is controlled by the $\sigma_1/\sigma_0$ ratio, growing much slower than $1/r^2$. Crucially, the non-vanishing bias term of the fixed decoder is eliminated by the generative decoder, resulting in a system that is only penalized by the sub-optimal bin widths, allowing for a large recovery of the Ideal gain.
\end{IEEEproof}

Lemma \ref{lemma:asymmetry} characterizes the \emph{efficiency} of recovery. In the resolution loss regime ($r < 1$), the performance gap is dominated by the fixed encoder's coarse granularity—an irreversible information loss that limits how close any decoder can get to the ideal benchmark. In contrast, in the tail mismatch regime ($r > 1$), the dominant error source is determining the scale of the outer bins. Since this is more of a calibration issue, the generative decoder successfully closes the gap, recovering a significant portion of the ideal gain.

\section{Robust Generative Decompression over Noisy Channels}
\label{sec:noisy}

In practical engineering systems, e.g., CSI feedback in wireless systems, the quantization index is transmitted over a noisy medium. We extend the framework to scenarios where the decoder observes a noisy version of the encoder output.

\subsection{Derivation of Soft-Decoding Rule}
Let $I = Q_d(X)$ be the transmitted index, and $\hat{I}$ be the received index after passing through a discrete memoryless channel characterized by transition probabilities $P(\hat{I}|I)$. The decoder seeks an estimator $g(\hat{I})$ minimizing the MSE under the true source distribution:
\begin{equation}
    \min_{g} \mathbb{E}_{t} \left[ (X - g(\hat{I}))^2 \right].
\end{equation}
The solution is the posterior mean given the \emph{received} symbol. By the law of total expectation, this expands to a soft-weighted combination of the generative centroids derived in Prop. \ref{prop:optimal-recon-true}:
\begin{align}
    a_{\hat{i}}^{t, \text{noisy}} &= \mathbb{E}_{t}[X \mid \hat{I} = \hat{i}] \nonumber \\
    &= \sum_{i=1}^N P(I=i \mid \hat{I}=\hat{i}) \cdot \underbrace{\mathbb{E}_{t}[X \mid I=i]}_{a_i^{t\star}}.
    \label{eq:soft_decoding}
\end{align}
Using Bayes' rule, the weighting term depends on both the channel statistics and the true source priors $p_i^{(t)} = \mathbb{P}_{t}(X \in \mathcal{R}_i)$:
\begin{equation}
    P(I=i \mid \hat{I}=\hat{i}) = \frac{P(\hat{I}=\hat{i} \mid I=i) p_i^{(t)}}{\sum_{k=1}^N P(\hat{I}=\hat{i} \mid I=k) p_k^{(t)}}.
\end{equation}

The optimal reconstruction rule $a_{\hat{i}}^{t, \text{noisy}}$ for the noisy channel is an MMSE soft estimator that operates on the posterior distribution $P(X \mid \hat{I}=\hat{i})$. This process breaks down into two distinct, cascading corrections: one for the source mismatch (generative shift of the centroids) and one for the channel error.

For the optimal generative shift of the centroid, the base reconstruction values must first be set to the true conditional means, $a_i^{t\star}$, derived using the true source PDF, $f_t(x)$, and the fixed bin boundaries, $\mathcal{R}_i$. This step ensures the decoder's dictionary is statistically accurate before channel uncertainty is considered.
\begin{equation}\label{eq:gencenteroids}
    a_i^{t\star} = \frac{\int_{\mathcal{R}_i} x f_t(x) dx}{\int_{\mathcal{R}_i} f_t(x) dx} = \mathbb{E}_{t}[X \mid X \in \mathcal{R}_i]
\end{equation}

Then, to deal with the channel error we use the received noisy index $\hat{i}$ to calculate the posterior probability $P(I=i \mid \hat{I}=\hat{i})$. This posterior is then used to linearly combine the generative centroids calculated in \eqref{eq:gencenteroids}, achieving the overall MMSE estimate
\begin{align}
    a_{\hat{i}}^{t, \text{noisy}} &= \sum_{i=1}^N P(I=i \mid \hat{I}=\hat{i}) \cdot a_i^{t\star},
\end{align}
which is equivalent to \eqref{eq:soft_decoding}.

This decoupled two stage correction shows that the benefits of generative decompression persist even when the communication link is imperfect. The optimal decoder gains robustness by integrating the corrected source priors ($a_i^{t\star}$) into the channel's soft decoding process. This is specifically beneficial for the semantic communication use cases~\cite{khosravirad2025ratelessjointsourcechannelcoding} where a context-aware decoder of a \gls{jscc} can improve on a distortion.

\subsection{Example Case: Binary Symmetric Channel (BSC)}
\label{sec:bsc-analysis}

Consider the illustrative  case of a 1-bit quantizer designed for $\mathcal{N}(0, \sigma_0^2)$ applied to a source $X \sim \mathcal{N}(0, \sigma_1^2)$, transmitted over a binary symmetric channel (BSC) with crossover probability $\epsilon$.

The generative centroids under the true distribution are $a_{0,1}^{t\star} = \mp \sigma_1 \sqrt{2/\pi}$. Due to symmetry, $p_0^{(t)} = p_1^{(t)} = 0.5$.
If the decoder receives $\hat{I}=1$, the posterior probabilities are $P(I=1|\hat{I}=1) = 1-\epsilon$ and $P(I=0|\hat{I}=1) = \epsilon$.

The decoder may take one of the following strategies towards the final estimate $\hat{x}$.

\paragraph{Standard separation (Unadjusted)}
Hard decision on $\hat{I}$, reconstruction with design centroids (based on $\sigma_0$):
\begin{equation}
\hat{x}_{\text{std}} = \operatorname{sgn}(\hat{I})\cdot\sigma_0\sqrt{\frac{2}{\pi}},
\end{equation}
which suffers from both bias $(\sigma_1-\sigma_0)$ and hard-decision loss.

\paragraph{Adjusted separation (Hard generative)}
Hard decision on $\hat{I}$, but using true generative centroids:
\begin{equation}
\hat{x}_{\text{hard}} = \operatorname{sgn}(\hat{I})\cdot\sigma_1\sqrt{\frac{2}{\pi}},
\end{equation}
which eliminates bias, but is still overconfident on noisy indices.

\paragraph{Optimal joint detection and decoding (Soft generative)}
The MMSE estimator softly combines the true centroids using channel posteriors:
\begin{align}
\hat{x}_{\text{opt}} &= (1-2\epsilon)\cdot\operatorname{sgn}(\hat{I})\cdot\sigma_1\sqrt{\frac{2}{\pi}}.
\end{align}
The factor $(1-2\epsilon)$ is a reliability-based shrinkage toward zero.

\textbf{Quantifying the separation penalty:}
Since $\hat{x}_{\text{opt}} = \mathbb{E}[X|\hat{I}]$ is the conditional mean, the mean squared error of \emph{any} other estimator $\hat{x}$ decomposes as:
\begin{equation}
    \mathbb{E}[(X - \hat{x})^2] = \underbrace{\mathbb{E}[(X - \hat{x}_{\text{opt}})^2]}_{D_{\text{opt}}} + \underbrace{\mathbb{E}[(\hat{x} - \hat{x}_{\text{opt}})^2]}_{\text{Excess Penalty}}.
\end{equation}
This allows us to quantify the distortion penalty of separation strategies by simply measuring their deviation from the optimal soft reconstruction.

The adjusted separation strategy correctly identifies the true centroid magnitude ($\sigma_1$) but fails to account for channel reliability, effectively forcing $\epsilon=0$. The deviation is:
\begin{align}
    \Delta_{\text{hard}} &= |\hat{x}_{\text{hard}} - \hat{x}_{\text{opt}}| \nonumber \\
    &= \left| 1 \cdot \sigma_1\sqrt{\frac{2}{\pi}} - (1-2\epsilon)\sigma_1\sqrt{\frac{2}{\pi}} \right| = 2\epsilon\sigma_1\sqrt{\frac{2}{\pi}}.
\end{align}
Squaring this yields the precise penalty imposed by the hard-decision interface:
\begin{equation}
    D_{\text{hard}} - D_{\text{opt}} = 4\epsilon^2 \sigma_1^2 \frac{2}{\pi}.
    \label{eq:separation-penalty}
\end{equation}

The standard separation strategy incurs an additional squared bias because it uses the design variance $\sigma_0$ rather than the true variance $\sigma_1$.
Using the estimator definitions $\hat{x}_{\text{std}} = \sigma_0 K$ and $\hat{x}_{\text{hard}} = \sigma_1 K$ (where $K=\text{sgn}(\hat{I})\sqrt{2/\pi}$), we can explicitly compare their excess distortions relative to the optimal $\hat{x}_{\text{opt}} = (1-2\epsilon)\sigma_1 K$:
\begin{align}
    D_{\text{std}} - D_{\text{opt}} &\propto (\sigma_0 - (1-2\epsilon)\sigma_1)^2, \\
    D_{\text{hard}} - D_{\text{opt}} &\propto (2\epsilon\sigma_1)^2.
\end{align}

In the limit of low noise ($\epsilon \to 0$), the standard penalty reduces to the source mismatch bias $(\sigma_0 - \sigma_1)^2$, while the hard-decision penalty vanishes. Consequently, for any non-trivial mismatch where $\sigma_0$ does not coincidentally act as a shrinkage factor, the distortion hierarchy is strictly:
\begin{equation}
    D_{\text{std}} > D_{\text{hard}} > D_{\text{opt}}.
\end{equation}
The gap $D_{\text{std}} - D_{\text{hard}}$ represents the \emph{source bias} penalty (solved by generative decompression), while $D_{\text{hard}} - D_{\text{opt}}$ represents the \emph{separation} penalty (solved by soft aggregation).

This closed-form analysis has direct implications for the source-channel separation theorem, which asserts that optimizing source and channel coding independently is asymptotically  optimal  when models are perfectly known \cite{csiszar1981information}. In our setting, however, the encoder is fixed and the source model is mismatched, so the \emph{operational performance} of a strictly modular, hard-decision separation architecture incurs an additional penalty under semantic mismatch~\cite{10747747} or coding mismatch~\cite{zhou2019dispersion}. Our result provides a precise quantification of this architectural penalty, yielding two key insights:

\begin{enumerate}
    \item {Channel reliability directly scales the magnitude:}
    The factor $(1-2\epsilon)$ in the optimal estimator represents the effective reliability of the received bit. When $\epsilon=0$, we recover the full generative centroids $\pm\sigma_1\sqrt{2/\pi}$. When $\epsilon=0.5$ (pure noise), the reconstruction collapses to zero. The optimal decoder effectively shrinks the estimate to minimize variance when the channel is unreliable.

    \item {Suboptimality of the separation architecture:}
    A strictly separate design first performs hard channel decoding and then applies a fixed reconstruction table. Such a decoder cannot produce the factor $\sigma_1(1-2\epsilon)$ because the interface was fixed assuming the wrong variance ($\sigma_0$) and the architecture enforces a hard decision. As shown in \eqref{eq:separation-penalty}, the separated decoder injects unnecessary power proportional to $\epsilon^2$, proving that the optimal reconstruction must inseparably couple the true source variance $\sigma_1$ with the channel reliability.
\end{enumerate}

In summary, even in this scalar setting, distribution mismatch exposes the inefficiency of a strictly modular separation architecture with a hard-decision interface. While distribution-aware centroids (hard generative strategy) eliminate bias, optimal performance requires \emph{joint source-channel awareness}—a feature that cannot be implemented in a strictly modular design with a frozen encoder, even though Shannon's separation theorem continues to hold when its modeling and joint-design assumptions are satisfied.

\paragraph{Benchmarking Against Ideal JSCC}
The ultimate benchmark is a system where the encoder is also re-optimized for both the true source variance $\sigma_1$ and the channel noise $\epsilon$.
For symmetric unimodal sources over a BSC, the optimal quantization threshold is the mean ($\mu=0$) regardless of the noise level or source variance \cite{kurtenbach1969quantizing}.
Since the fixed encoder already uses this threshold, our proposed \emph{soft generative} decoder achieves the global lower bound for \gls{jscc}:
\begin{equation}
    D_{\text{opt}} = D_{\text{ideal}}^{\text{JSCC}}.
\end{equation}
This implies that for binary signaling, fixing the encoder imposes zero penalty---all gains are recoverable at the decoder. However, for $N > 1$, an ideal JSCC encoder would optimize bit-mappings (index assignment) to minimize Hamming distance errors—a degree of freedom unavailable to the fixed encoder—meaning $D_{\text{opt}}$ would eventually diverge from $D_{\text{ideal}}^{\text{JSCC}}$ at higher rates.

This optimality result relies on the fact that a 1-bit system has trivial topology: a bit flip always transitions to the unique alternative centroid. For multi-bit systems ($N > 1$), the relationship between the generative decoder and the ideal benchmark is governed by the efficiency of the encoder. The positive gaps $D_{\text{gen}} > D_{\text{ideal}}$ re-emerge for $N > 1$, with the magnitude of the gap determined by the inefficiency of the fixed encoder's index assignment relative to the channel noise profile. Hence, the conclusion that fixing the encoder imposes zero penalty applies only within this scalar binary-signaling setting.

\section{Task-Aware Decoding Under Mismatch}
\label{sec:taskaware}

The framework developed in previous sections focused on MSE distortion. In many modern applications, however, the decoder is not judged by reconstruction fidelity alone but by a \emph{downstream task} that depends on the received information. In such settings, the mismatch between design and true source distributions persists, but the relevant distortion measure differs from or is at odds with MSE~\cite{blau2018perception}.

In this section, we extend the mismatched quantization formulation to general task-aware distortion measures. We show that the optimal decoder principle holds—the decoder should optimize against the true distribution—but the optimal reconstruction rule changes from a conditional mean to a task-specific estimator.

\subsection{General Task-Aware Decoder Optimization}

Let $Q_d$ denote a fixed encoder designed under the assumed distribution $f_X(\cdot;\theta_d)$, producing index $I=Q_d(X)$. Let $d_{\mathsf{task}}(x,\hat{x})$ be any task-specific distortion measure. The decoder seeks a mapping $g : \{1,\dots,N\} \to \mathbb{R}$ that minimizes the expected task loss under the \emph{true} distribution:
\begin{equation}
D_{\mathsf{task}}(g) = \mathbb{E}_{f_X(\cdot;\theta_t)} \big[ d_{\mathsf{task}}( X , g(I) ) \big].
\label{eq:general-task-loss}
\end{equation}
Conditioning on index $I=i$, the optimal decoder solves $N$ independent scalar minimizations:
\begin{equation}
g^*(i) = \arg\min_{a\in\mathbb{R}} \mathbb{E}_{\theta_t}\big[d_{\mathsf{task}}(X,a)\mid X\in\mathcal{R}_i\big].
\label{eq:task-optimal-reconstruction}
\end{equation}

\begin{figure*}[t]
    \centering
    \includegraphics[width=0.75\textwidth]{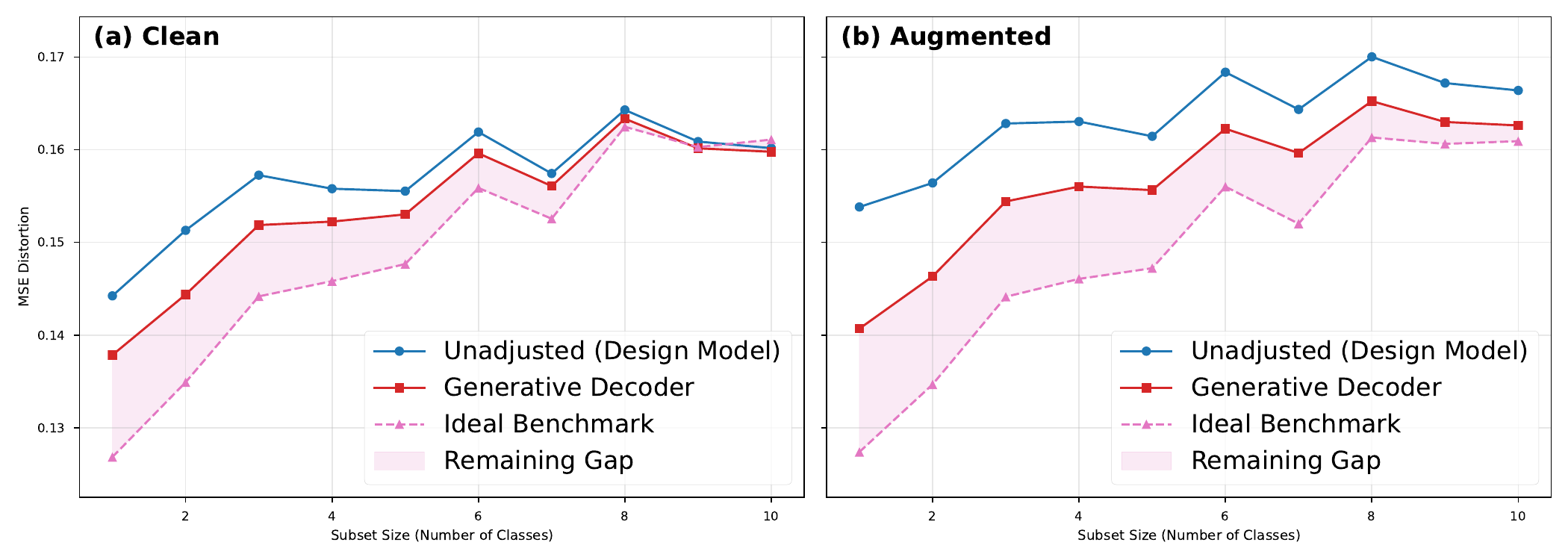}
    \caption{MSE distortion for 5 bits ($b=5$), comparing clean and augmented scenarios. The generative decoder closes a significant portion of the gap to the ideal benchmark despite the frozen encoder.}
    \label{fig:distortion_gap}
\end{figure*}

\subsection{CSI Compression for Beamforming and Site-Specific Decoder Adaptation}
\label{sec:csi-type2}

In this subsection, we extend the analysis to complex-valued channel coefficients to model wireless fading channels more accurately.

5G \gls{nr} CSI feedback employs a fixed, standardized linear-combination codebook designed to be robust for a wide range of multi-path and spatially correlated channel distributions. In real deployments, however, channels are highly site-specific. These deviations are known only to the gNB through long-term measurements, not to the UE encoder. While mathematical modeling of goal-oriented distortion over spatially correlated multi-path channels is generally intractable, it is valuable to derive analytical insights into the possible gain from a generative decoder using a tractable scalar proxy.

Recent work \cite{luo2025generative} demonstrated that freezing the standardized UE encoder and only retraining the gNB-side CSI refinement network on site-specific data reduces feedback overhead significantly at equal sum-rate. This serves as a large-scale empirical validation of generative decompression.

In the following, we provide analytical insight into the dominant mechanism driving these gains by modeling the critical quantity in CSI feedback: the complex coefficient of the dominant beam.

\subsubsection{Task Loss Model}
A simplified yet effective proxy for the loss in downlink beamforming gain (and hence sum-rate) due to imperfect CSI is the weighted MSE:
\begin{equation}
    d_{\mathsf{CSI}}(x,\hat{x}) = |x|^2 \, |x - \hat{x}|^2.
    \label{eq:CSI-loss}
\end{equation}
This penalizes reconstruction errors more heavily when the channel coefficient magnitude $|x|$ is large—i.e., the regime that contributes most to the received SNR in massive MIMO \cite{jindal2006mimo,bjornson2017massive}.
It is important to note that this scalar analysis primarily captures the accuracy of \emph{gain estimation} (power allocation) rather than the spatial direction (phase alignment) of the precoder, which relies on high-dimensional eigenvector quantization. However, since the dominant beam coefficient dictates the effective channel gain, minimizing the weighted magnitude error is a necessary condition for maximizing sum-rate.

\subsubsection{Site-Specific Model with Rician Dominant Beam Coefficient}
We model the dominant beam coefficient as complex Rician. Without loss of generality, we rotate the phase so that the \gls{los} component lies on the real axis. The coefficient $X$ is then
\begin{equation}
X \sim \mathcal{CN}\left( \sqrt{\frac{K}{K+1}},\ \frac{1}{K+1} \right),
\end{equation}
with unit average power $\mathbb{E}[|X|^2] = 1$.

The design quantizer is optimized for typical scenarios with low-to-moderate effective $K_d \approx 1$--$3$ in codebook design. The true, site-specific, channel has $K_t$ that can deviate from the design, e.g., $K_t > K_d$ (i.e., stronger LoS) which can be experienced in urban macro scenarios, or $K_t < K_d$ expected in more rich scattering scenarios.

\subsubsection{Optimal vs. Fixed Decoder Analysis}
To quantify the analytical gain, we define the reconstruction functions for the fixed (standardized) and adaptive (generative) strategies.

Let $\mathcal{M}_n(K) = \mathbb{E}[X^n \mid \mathrm{Re}(X)>0; K]$ denote the $n$-th conditional moment of the Rician distribution with parameter $K$. The task-optimal reconstruction centroid as a function of $K$ is:
\begin{equation}
    \phi(K) \triangleq \frac{\mathcal{M}_3(K)}{\mathcal{M}_2(K)}.
\end{equation}
We analyze the system mismatch where the standard decoder uses a fixed centroid $a_{\text{fix}} = \phi(K_d)$ optimized for the design assumption $K_d$, while the generative decoder uses the true optimal centroid $a^\star = \phi(K_t)$.

The relative performance gain is defined as:
\begin{equation}
    \eta(K_t) = 1 - \frac{\mathbb{E}_{K_t}[|X|^2 |X - \phi(K_t)|^2]}{\mathbb{E}_{K_t}[|X|^2 |X - \phi(K_d)|^2]}.
\end{equation}

We examine three regimes of the true channel $K_t$:

\paragraph{Rich Scattering Mismatch ($K_t \rightarrow 0$)}
Here, the site is rich-scattering (Rayleigh), but the codebook was designed for a dominant beam ($K_d>0$). The normalized true distribution ($K_t=0$) has moments $\mathcal{M}_2(K_t)=1$, $\mathcal{M}_3(K_t) = \kappa = 2\sqrt{2/\pi}$ and $\mathcal{M}_4(K_t)=3$.
    The weighted distortion function simplifies to a parabola:
    \begin{equation}
        D(a) =  3 - 2\kappa a + a^2.
    \end{equation}
The adaptive decoder selects the minimum of this parabola, $a^\star = \kappa$, while the fixed decoder is locked to the centroid of the design distribution, $a_{\text{fix}} = \phi(K_d) < \kappa$. Note that $\phi(K)$ is a decreasing function. The excess distortion caused by the fixed decoder is exactly then
\begin{equation}
    D(a_{\text{fix}}) - D(a^\star) = (\kappa - \phi(K_d))^2,
\end{equation}
Which provides the relative gain given as:
    \begin{equation}
        \eta(K_d) = \frac{(\kappa - \phi(K_d))^2}{3 - \kappa^2 + (\kappa - \phi(K_d))^2}.
    \end{equation}
   
The gain is driven purely by the rigidity of the fixed decoder. For a typical design of $K_d=3$, we have $\phi(3) \approx 1.3$, resulting in $ \eta(K_d) \approx 18 \%$.    
The adaptive decoder biases the reconstruction outward to capture the heavy tail of the Rayleigh distribution, yielding a constant relative gain factor derived from the ratio of Rician to Rayleigh tail integrals.

\paragraph{Mild Mismatch ($K_t \approx K_d$)}
Consider a moderate deviation where the true site has $K_t=6$ (stronger LoS) while the design assumes $K_d=3$.
Unlike the rich scattering case, the optimal centroid $\phi(K)$ is relatively stable in this region, where $\eta(6)$ yields a relative distortion reduction of approximately $8\%$, although this correction comes at no cost and accumulates across many subbands.

\paragraph{LoS Mismatch ($K_t \to \infty$)}
This is the case where a user has a pristine LoS path. As $K_t$ increases, the distribution $X$ converges to a Dirac delta at $\mu=1$ (unit power). The optimal reconstruction converges to the signal itself,
    \begin{equation}
        \lim_{K_t \to \infty} a^\star = \lim_{K_t \to \infty} \phi(K_t) = 1,
    \end{equation}
and, the adaptive distortion approaches zero: $D_{\text{opt}} \to |1|^2 |1 - 1|^2 = 0$. The fixed decoder remains anchored to $K_d$ and reconstructs $a_{\text{fix}} = \phi(K_d) > 1$. Consequently, $ \lim_{K_t \to \infty} \eta(K_t) = 100\%$. 

The optimal generative decoder automatically adapts to site-specific LoS strength: it exaggerates strong coefficients in rich-scattering/low-$K$ cases to maximize beamforming gain, and confidently collapses to the LoS value when $K_t$ is high. A conventional fidelity-oriented decoder cannot exhibit this dual behavior---it remains conservative because it is calibrated for the average case. This simple analytical model therefore provides the theoretical explanation for why purely decoder-side, site-specific refinement (whether closed-form or learned) yields such dramatic feedback savings in practice.

\begin{figure}[t]
    \centering
    \includegraphics[width=0.49\textwidth]{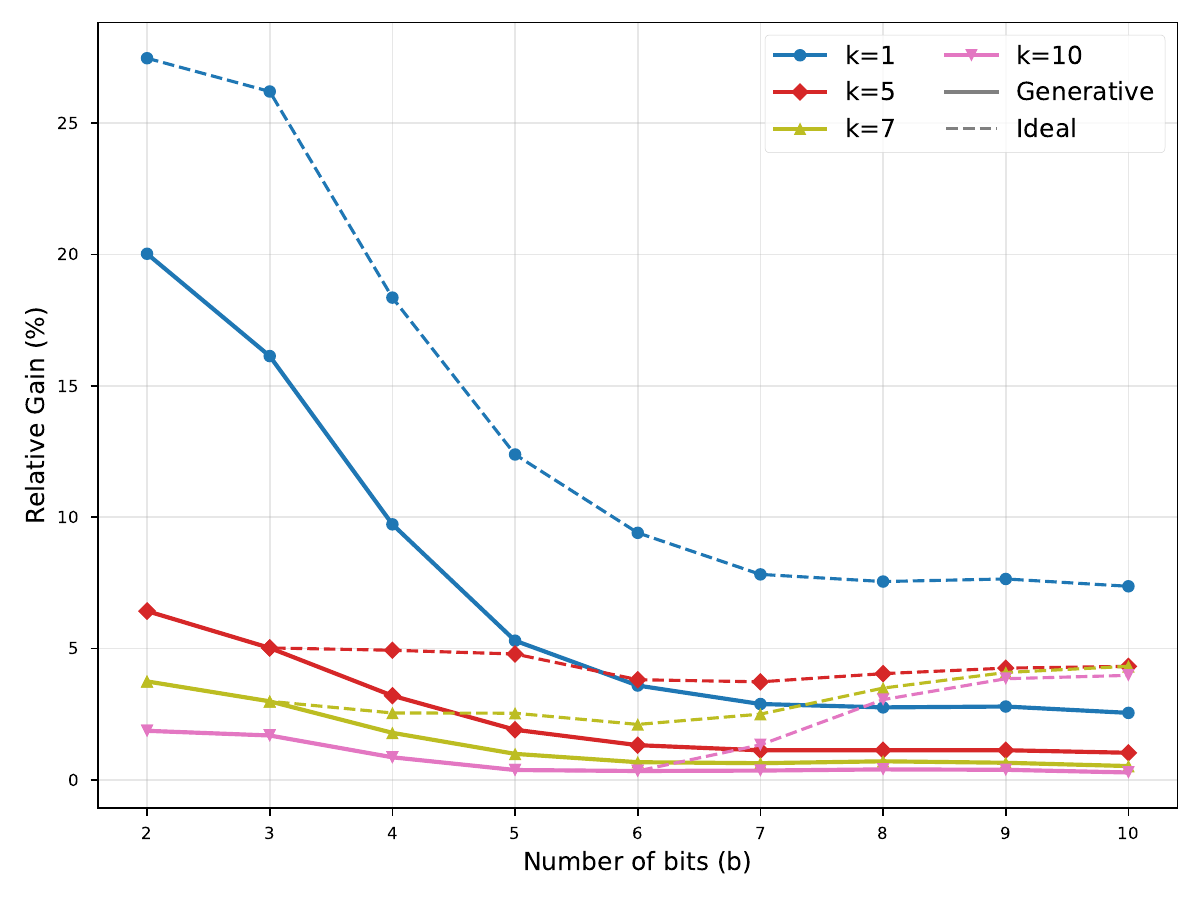}
    \caption{Relative reconstruction gain vs.\ codebook size (bits) for varying target subset sizes $k$. Adaptation gains are highest for severe mismatch (small $k$) and low bit rates.}
    \label{fig:gain_vs_bits_clean}
\end{figure}

\begin{figure}[ht]
    \centering
    \includegraphics[width=0.48\textwidth]{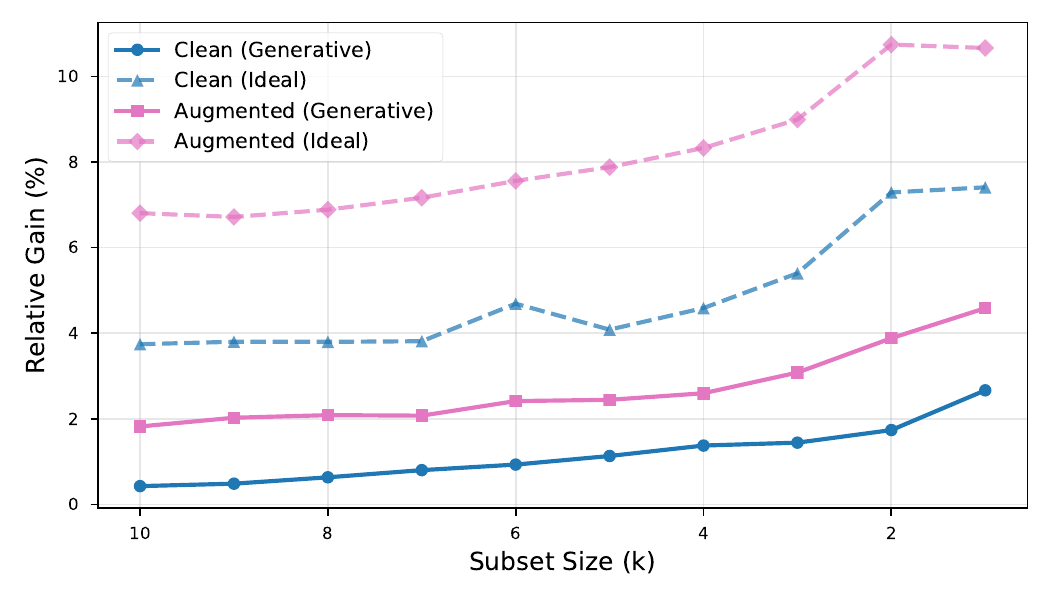}
    \caption{Reconstruction gain comparison between clean and augmented scenarios at 8 bits. The relative benefit of generative adaptation is robust to domain shifts (augmentation).}
    \label{fig:clean_vs_augmented}
\end{figure}

\subsection{Context-Aware Semantic Decoding}
\label{sec:semantic-decoding}

We now examine decoder-side adaptation when the receiver is interested in downstream semantic utility (classification accuracy) alongside signal fidelity.

\begin{figure*}[t]
    \centering
    \includegraphics[width=0.75\textwidth]{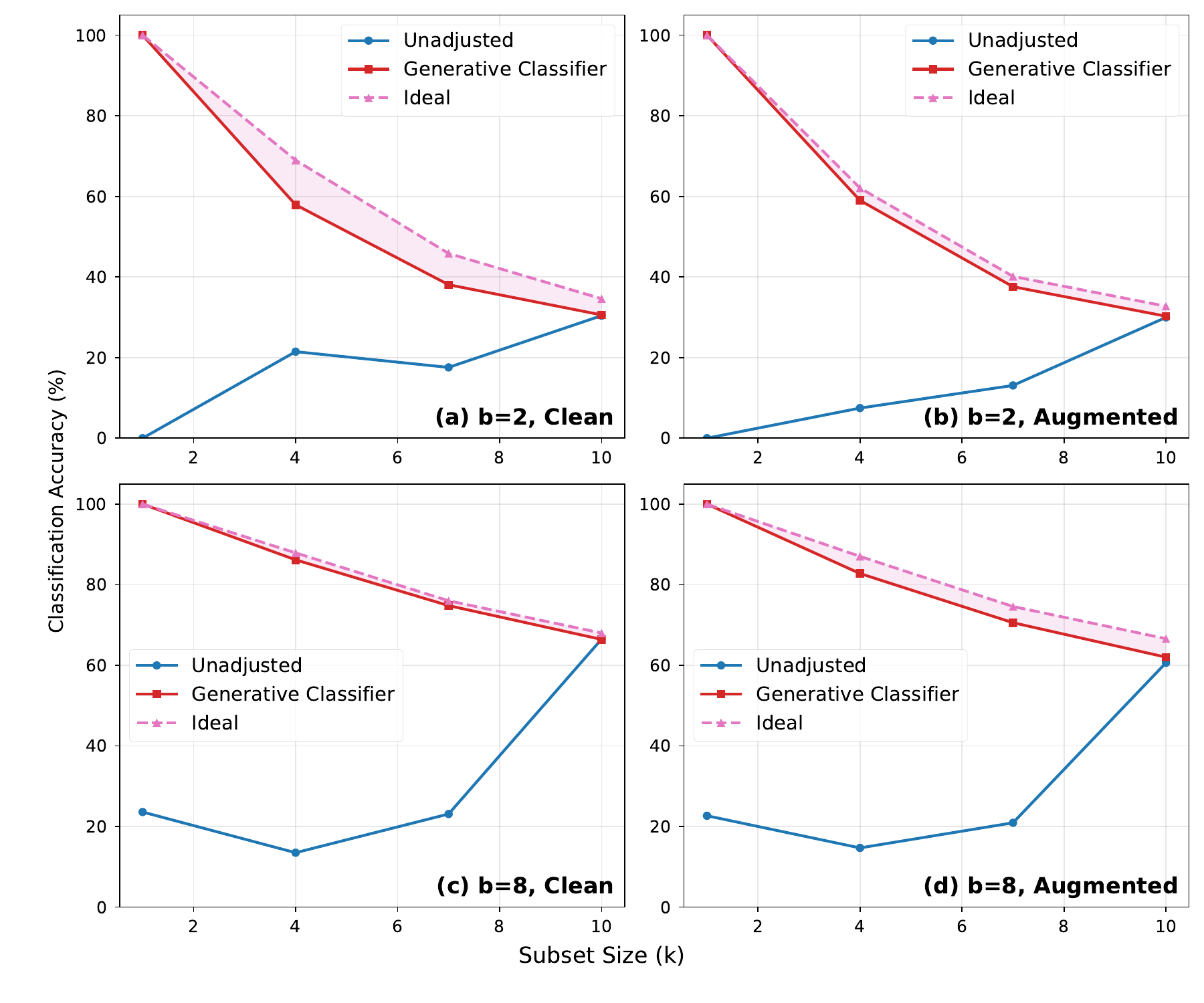}
    \caption{Classification accuracy vs.\ subset size $k$ for varying bit depths. The unadjusted baseline collapses for small $k$, while the generative classifier recovers significant accuracy.}
    \label{fig:accuracy_vs_subset}
\end{figure*}

\subsubsection{The Semantic Decision Rule}
While Section II established that the conditional mean is optimal for minimizing MSE, minimizing semantic misclassification requires a distinct decision strategy.

Let $Y \in \{1, \dots, M\}$ be the semantic label for source $X$. The decoder observes index $I = Q_d(X)$ and seeks a decision $\hat{y} = g(I)$ to minimize the error probability $P(\hat{Y} \neq Y)$.
It is a standard result in decision theory that the optimal rule is the \gls{map} estimator. In the context of generative decompression, this requires the decoder to evaluate the posterior under the \emph{true} distribution $P_t$:
\begin{equation}
    \hat{y}^*(i) = \arg \max_{y \in \{1,\dots,M\}} \mathbb{P}_t(Y=y \mid X \in \mathcal{R}_i).
    \label{eq:sem-optimal}
\end{equation}
This highlights a fundamental operational shift: whereas fidelity-oriented decoding (MSE) computes a \emph{weighted average} (Prop.~\ref{prop:optimal-recon-true}), task-oriented decoding performs a \emph{hard selection} based on the dominant class mass within the quantization bin $\mathcal{R}_i$.
Consequently, generative adaptation in the semantic regime does not strictly mean adjusting centroids; rather, it entails re-estimating the class priors $\mathbb{P}_t(Y|I)$ within each fixed bin—effectively what is achieved by retraining the classifier head on the target distribution.

\begin{figure*}[t]
    \centering
    \includegraphics[width=0.75\textwidth]{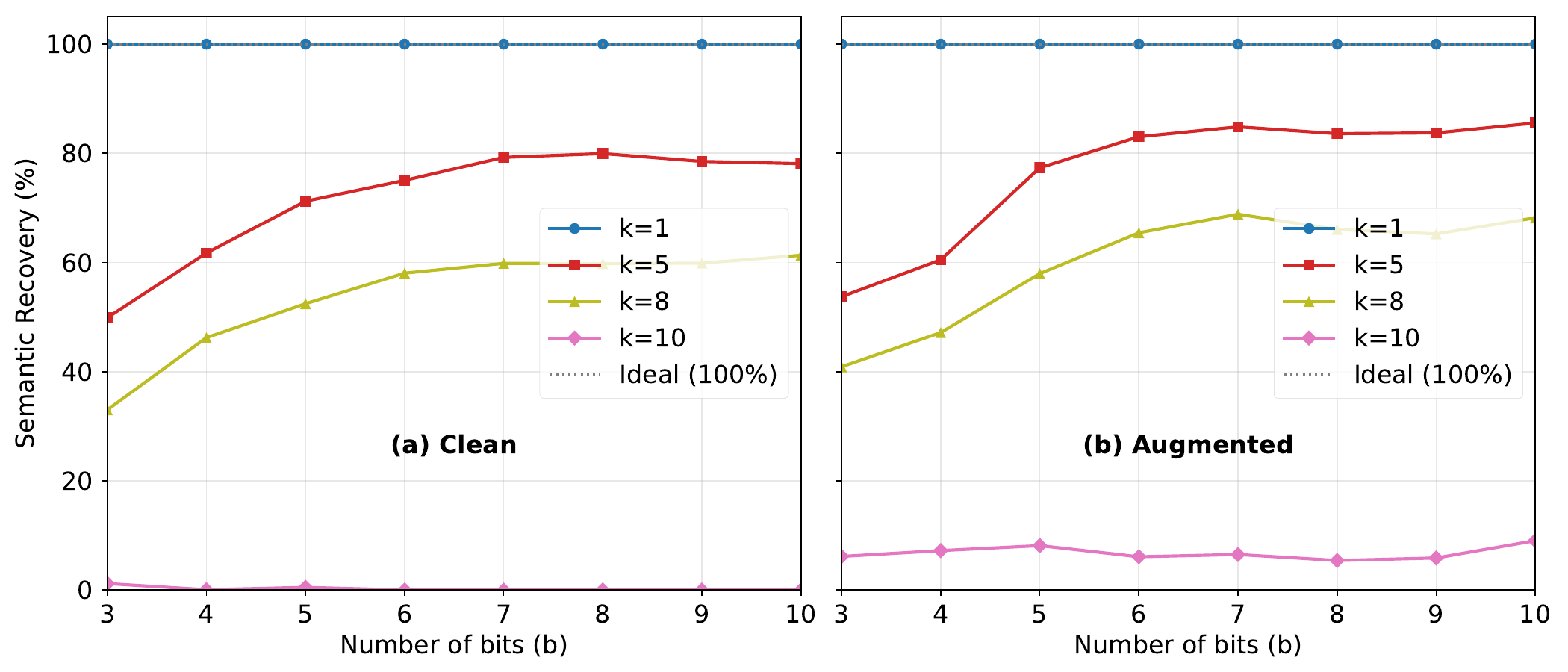}
    \caption{Semantic recovery (\%) vs.\ codebook size ($b$). Recovery is poor at low rates ($b=2$) due to the frozen encoder's structural bottleneck but improves significantly as codebook capacity increases.}
    \label{fig:semantic_gain_vs_bits}
\end{figure*}

\subsubsection{Experimental Setup}
To validate the generative decompression principle in a high-dimensional setting, we employ a VQ-VAE architecture on the Fashion-MNIST dataset~\cite{xiao2017fashionmnist}. We simulate \emph{semantic robustness} scenarios where the encoder is trained on a global mixture, but the decoder faces shifted distributions unknown to the transmitter. Throughout this section we distinguish between an \emph{Ideal Benchmark}, in which both encoder (including the codebook) and decoder are retrained end-to-end on the target distribution, and a \emph{Generative Decoder} regime, in which the encoder is frozen (pre-trained on the source distribution) and only the decoder parameters are adapted.

\textbf{Architecture \& Mismatch:} A fixed VQ-VAE maps images $X$ to latent vectors $z \in \mathbb{R}^{64}$, quantized to $N=2^b$ codebook entries ($b \in \{2,\dots,10\}$). Mismatch is induced via:
\begin{itemize}
    \item \emph{Class Restriction:} Only $k$ out of 10 classes appear at test time ($k \in \{1,\dots,10\}$), simulating a specialized deployment environment.
    \item \emph{Domain Shift:} Images undergo geometric and photometric augmentations (noise, rotation, scaling) not seen during encoder training.
\end{itemize}

\textbf{Adaptation Regimes:} We compare three strategies:
\begin{enumerate}
    \item \emph{Unadjusted Decoder:} The design model, trained on the full mixture, is evaluated directly on the target distribution.
    \item \emph{Generative Decoder:} The encoder uses the frozen design codebook $\mathcal{C}_{\text{enc}}$ to generate indices (fixing the partitions $\{\mathcal{R}_i\}$). The decoder utilizes a separate, decoupled codebook $\mathcal{C}_{\text{dec}}$. During adaptation, $\mathcal{C}_{\text{dec}}$ and the subsequent heads are fine-tuned via backpropagation on the target distribution to minimize MSE or classification error, while $\mathcal{C}_{\text{enc}}$ remains frozen.
    \item \emph{Ideal Benchmark:} The entire pipeline (encoder and decoder) is jointly re-optimized, providing an upper bound on performance.
\end{enumerate}

\textbf{Metrics:}
For subset size $k$ and bit depth $b$, we measure MSE distortion $D$ and classification accuracy $A$. The relative effectiveness of the proposed adaptation is quantified by the fraction of the \emph{ideal} gain recovered:
\begin{equation}
\label{eq:semantic_recovery}
\text{Semantic Recovery} = \frac{A_{\text{gen}} - A_{\text{fix}}}{A_{\text{ideal}} - A_{\text{fix}}}.
\end{equation}

\subsubsection{Experimental Results}

\paragraph{Reconstruction Fidelity}
We first analyze the signal fidelity capabilities of the generative decoder. Figure~\ref{fig:distortion_gap} visualizes the raw MSE distortion across the three regimes. The unadjusted decoder suffers from a systematic bias that does not vanish with rate, consistent with the theoretical prediction of a distortion floor under mismatch. The generative decoder lowers this floor, closing the majority of the gap to the ideal benchmark without modifying the encoder.

The relationship between codebook capacity and adaptation gain is detailed in Fig.~\ref{fig:gain_vs_bits_clean}. The relative gain decreases monotonically with bit depth $b$. This trend suggests that as $b$ decreases, and thus the distortion increases, there is more room for the generative decoder to repurpose codewords that are unused by the restricted class distribution, effectively increasing the resolution for the active classes. The ideal benchmark performs better in higher rate where adjusting the codewords on the encoder side provide the extra gain.

Figure~\ref{fig:clean_vs_augmented} demonstrates the robustness of this approach under domain shifts. While augmented images naturally incur higher absolute distortion, the \emph{relative} recovery remains consistent with the clean scenario. This confirms that the Bayesian correction principle (Section~\ref{sec:problem}) is robust to complex, non-Gaussian shifts; even when the source is perturbed by noise or rotation, the decoder-side adaptation effectively adjusts the quantization rule.

\paragraph{Semantic Inference}
We now examine the performance when the decoder maximizes a downstream utility function (accuracy) rather than fidelity.

Figure~\ref{fig:accuracy_vs_subset} presents the absolute classification accuracy as a function of the subset size $k$. A distinct performance collapse is observed for the unadjusted decoder (blue curves) as the subset size decreases. This is expected, as the fixed classifier operates under the prior assumption of a balanced 10-class mixture, leading to false positives for classes not present in the target $k$-subset. The generative classifier (red curves) mitigates this by effectively re-estimating the class priors $P(Y|\hat{I})$ within each quantization bin. While a gap remains compared to the ideal benchmark (magenta curves), the generative adaptation recovers a critical portion of the usable accuracy.

To quantify this recovery, Figure~\ref{fig:semantic_gain_vs_bits} plots the semantic recovery metric defined in \eqref{eq:semantic_recovery}. Semantic recovery exhibits a strong dependence on both bit depth and subset size. At low rates ($b=2$), recovery is notably poor. This represents a \emph{structural bottleneck}: the partitions determined by the frozen encoder $\mathcal{C}_{\text{enc}}$ are too coarse to separate classes, fusing samples from different labels into the same bin. In this regime, no decoder-side re-weighting can disentangle the classes. However, as bit depth increases ($b \ge 4$), the encoder partitions become fine enough that the decoder's updated priors can effectively distinguish the active classes, leading to high recovery rates. Expectedly, recovery is almost at $0\%$ at $k=10$,  since the generative decoder has no room to improve upon the unadjusted classifier given the uniform distribution over classes.

\section{Conclusion}
\label{sec:conclusion}
This work establishes a theoretical foundation for \emph{generative decompression}, a paradigm shifting the burden of adaptation from the resource-constrained or standard-constrained transmitter to the intelligent receiver. We proved that when the encoder is fixed or mismatched—a common constraint in standardized cellular protocols—the optimal decoding strategy is a Bayesian correction that aligns reconstruction with the true source distribution.

Our analysis yields three critical insights for future communication systems. First, we showed that distribution mismatch acts as a distortion floor, which can be reduced solely through improved decoder-side priors. Second, we rigorously quantified the inefficiency of standard modular source-channel separation architectures under mismatch with a fixed encoder, deriving a robust soft-decoding rule that connects classical quantization and modern joint source-channel coding (JSCC) while remaining consistent with Shannon's separation theorem under its usual assumptions. Third, the efficacy of this framework in high-dimensional tasks, such as semantic classification and CSI feedback compression, suggests that generative decompression is a key enabler for \emph{context-aware communications}. As 6G systems move toward AI-native interfaces, this framework provides the analytical justification for deploying powerful, generative models at the receiver to compensate for coarse, standardized transmission.

\bibliographystyle{IEEEtran}
\bibliography{references}

\end{document}